\documentclass[a4paper,11pt,fleqn]{article}

\usepackage{amsmath, amssymb, amsthm}
\usepackage{mathtools, thm-restate}
\usepackage{mathrsfs}
\usepackage[margin=1in]{geometry}
\usepackage{xcolor}
\usepackage{url}
\usepackage{comment}
\usepackage{tikz}
\usepackage{enumitem}
\usepackage{array}
\usepackage{longtable}
\usepackage{diagbox}
\usepackage{anyfontsize}
\usepackage{xfrac}
\usepackage[justification=centering]{caption}
\usepackage[hypertexnames=false,final]{hyperref}
\usepackage{algorithm, algpseudocode}
\usepackage[capitalize,sort]{cleveref}
\usetikzlibrary{decorations.pathreplacing,arrows.meta,calc}

\usepackage{amsmath}

\let\epsilon\varepsilon

\newcommand*{\thmdep}[2]{}

\newcommand*{\Th}{^{\textrm{th}}}

\newcommand*{\wLoG}{without loss of generality}

\let\eps\epsilon

\newcommand*{\lavec}[1]{\mathbf{#1}}  %

\newcommand*{\defeq}{:=}

\newcommand*{\ceil}[1]{\left\lceil #1 \right\rceil}

\DeclareMathOperator*{\E}{E}

\DeclareMathOperator*{\argmax}{argmax}

\DeclareMathOperator*{\opt}{opt}

\DeclareMathOperator{\Sum}{sum}

\DeclareMathOperator{\LP}{LP}
\DeclareMathOperator{\IP}{IP}

\usepackage{iftex}
\usepackage{url}
\usepackage{mathtools}
\usepackage{comment}

\hypersetup{
    colorlinks,
    linkcolor={red!50!black},
    citecolor={red!50!black},
    urlcolor={blue!50!black}
}

\makeatletter
\g@addto@macro{\UrlBreaks}{%
\do\/%
\do\a\do\b\do\c\do\d\do\e\do\f\do\g\do\h\do\i\do\j\do\k\do\l\do\m%
\do\n\do\o\do\p\do\q\do\r\do\s\do\t\do\u\do\v\do\w\do\x\do\y\do\z%
\do\A\do\B\do\C\do\D\do\E\do\F\do\G\do\H\do\I\do\J\do\K\do\L\do\M%
\do\N\do\O\do\P\do\Q\do\R\do\S\do\T\do\U\do\V\do\W\do\X\do\Y\do\Z%
\do\0\do\1\do\2\do\3\do\4\do\5\do\6\do\7\do\8\do\9%
}
\makeatother

\renewcommand*{\defeq}{\coloneqq}
\newcolumntype{L}{>{$\displaystyle}l<{$}}
\algnewcommand{\LineComment}[1]{\State \textcolor{gray}{// #1}}

\newcommand*{\acknowledgements}[1]{\paragraph{Acknowledgements.} #1}

\newtheorem{theorem}{Theorem}
\newtheorem{definition}{Definition}

\newtheorem{corollary}{Corollary}[theorem]
\newtheorem{lemma}[theorem]{Lemma}

\newtheorem{observation}[theorem]{Observation}

\crefname{claim}{Claim}{Claims}
\crefname{property}{Property}{Properties}
\crefname{observation}{Observation}{Observations}
\crefname{transformation}{Transformation}{Transformations}

\title{Analysis of the Harmonic Function Used in Bin-Packing}
\author{Eklavya Sharma\\
Department of Computer Science and Automation\\
Indian Institute of Science, Bengaluru.\\
\texttt{eklavyas@iisc.ac.in}}
\date{\empty}

\DeclareMathOperator{\score}{score}
\DeclareMathOperator{\cost}{cost}
\newcommand*{\zhat}{\widehat{z}}
\newcommand*{\hdhk}{\operatorname{\mathtt{HDH}}_k}

\begin{document}

\maketitle

\begin{abstract}
The harmonic function was first introduced by Lee and Lee \cite{leelee}
for analyzing their online bin-packing algorithm.
Subsequently, it has been used to obtain approximation algorithms
for many different packing problems.
Here we slightly generalize the harmonic function
and give alternative proofs of its important properties.

\end{abstract}

\setlength{\parskip}{0.6em}
\setlength{\parindent}{0pt}

\acknowledgements{I am grateful to my advisor, Prof.~Arindam Khan, for his valuable comments.}

\section{Introduction}

Bin-packing is one of the most important problems in operations research
and combinatorial optimization.
Several approximation algorithms have been devised for this problem,
in both the online and offline setting \cite{coffman2013bin}.
Geometric bin-packing, a natural variant of the bin-packing problem,
has also received significant attention,
and many approximation algorithms exist for it \cite{CKPT17}.

Bin-packing is NP-hard, so we seek approximation algorithms.
The worst-case approximation ratio usually occurs only for small pathological instances.
Thus, the standard performance measure is the asymptotic approximation ratio (AAR).

Lee and Lee gave a simple approximation algorithm for online bin-packing \cite{leelee}.
Their algorithm is parametrized by an integer $k \ge 2$ and uses at most $T_k\opt(I) + k$ bins
to pack items $I$ ($\opt(I)$ is the minimum number of bins needed to pack $I$),
where $T_k$ is a function of $k$ and $T_{\infty} \defeq \lim_{k \to \infty} T_k \approx 1.69103$.
The analysis of their algorithm uses a novel technique called \emph{harmonic function}.
They also prove that their algorithm's AAR is optimal for $O(1)$-space online algorithms.

Subsequently, many harmonic-based algorithms were given for online bin-packing
\cite{ramanan1989line,seiden2002online,BaloghBDEL18},
geometric bin-packing \cite{caprara2008,EpsteinS05,CsirikV93,han2011new,rna}
and geometric strip-packing \cite{caprara2008,BansalHISZ13}.
Caprara's $\hdhk$ algorithm \cite{caprara2008} for $d$-dimensional geometric bin-packing ($d$D BP)
has an AAR of $T_k^{d-1}$ (where $k$ is a parameter to the algorithm),
which is the best-known AAR for $d \ge 3$.

The crucial ingredient in Lee and Lee's analysis \cite{leelee}
is a closed-form expression for the max-knapsack-profit of the harmonic function
(we formally define max-knapsack-profit of a function later).
This result forms the basis of many harmonic-based algorithms,
including Lee and Lee's algorithm.
Caprara's algorithm for $d$D BP also relies on this result,
but they use a slightly different variant of the harmonic function.
They state a closed-form expression for the max-knapsack-profit of their variant,
but they don't prove it (probably because the proof is very similar to that of Lee and Lee).
We fill this gap by giving a generalized definition of the harmonic function
that captures both of these variants and obtaining a closed-form expression
for the max-knapsack-profit.

Our proof is also different from that of Lee and Lee:
we use integer linear programs in our proof, which makes the proof simpler.
Our proof is also more detailed.

\subsection{Preliminaries}

For $n \in \mathbb{Z}_{\ge 0}$, let $[n] \defeq \{1, 2, \ldots, n\}$.

Let $X \defeq [x_1, x_2, \ldots, x_n]$ be a sequence of $n$ numbers where $x_i \in [0, 1]$.
For any function $f: [0, 1] \mapsto [0, 1]$, define $f(X) \defeq [f(x_i): i \in [n]]$
and define $\Sum(X) \defeq \sum_{i=1}^n x_i$.

\begin{definition}
For a function $f: [0, 1] \mapsto [0, 1]$, the max-knapsack-profit is defined as
the maximum value of $\Sum(f(X))$ over all sequences $X$ where $\Sum(X) \le 1$.
Equivalently, this is the maximum profit attainable when we have to pack items
into a knapsack of size 1 and each item of size $x$ has profit $f(x)$.
\end{definition}

\subsection{Formal Statement of Results}

Let $k \in \mathbb{Z}_{\ge 1}$ and $\mu \in [0, k]$ be constants.
The harmonic function $f_k: [0, 1] \mapsto [0, 1]$ is defined as
\[ f_k(x) = \begin{cases}
{\displaystyle \frac{1}{j}} & {\displaystyle x \in \left(\frac{1}{j+1}, \frac{1}{j}\right]}
    \; \forall j \in [k-1]
\\[10pt] \mu x & {\displaystyle x \in \left[0, \frac{1}{k}\right]}
\end{cases} \]
Lee and Lee \cite{leelee} use $\mu = k/(k-1)$ and Caprara \cite{caprara2008} uses $\mu = k/(k-2)$.

Our aim is to find the max-knapsack-profit of $f_k$.
We can express this as an optimization problem $\LP(k, \mu)$:
\[ \sup_{n > 0} \left( \sup_{X \in \mathbb{R}^n_{\ge 0}} \Sum(f_k(X))
\textrm{ where } \Sum(X) \le 1 \right) \]

Our first important result states that we can find $\opt(\LP(k, \mu))$ by
solving an integer linear program.
\begin{theorem}
\label{thm:ip-eq-lp}
Let $\IP(k, \mu)$ be the following integer linear program:
\[ \min_{z \in \mathbb{Z}_{\ge 0}^{k-1}}
\mu + \sum_{j=1}^{k-1} z_j\left(\frac{1}{j} - \frac{\mu}{j+1}\right)
\textrm{ where } \sum_{j=1}^{k-1} \frac{z_j}{j+1} < 1 \]
Then $\opt(\LP(k, \mu)) = \opt(\IP(k, \mu))$.
\end{theorem}

This gives a simple algorithm for computing $\opt(\LP(k, \mu))$:
simply iterate over all the integral solutions to $\IP(k, \mu)$ and report the optimal one.
We will show that this can be done in $O(k!)$ time.

Our second important result is finding a closed-form expression for $\opt(\IP(k, \mu))$.
\begin{theorem}
\label{thm:exact}
Let $k \ge 2$ and $1 \le \mu < 2$. Let
\[ r_j \defeq \begin{cases} 1 & j = 1
\\ r_{j-1}(r_{j-1}+1) & j > 1 \end{cases} \]
and
\[ Q \defeq \argmax_{j \ge 1} \left(r_j \le \ceil{\frac{1}{\max(\mu-1, 1/k)}}-1\right). \]
Then
\[ \opt(\IP(k, \mu)) = \sum_{j=1}^{Q+1} \frac{1}{r_j} + \frac{\mu-1}{r_{Q+1}}. \]
\end{theorem}
It is easy to prove that $\opt(\LP(k, \mu)) = \mu$ when $\mu \ge 2$
(see \cref{thm:simple-bound}) or $k=1$.

Our third important result is about the limiting behavior of $\LP(k, \mu)$
when $k \to \infty$ and $\mu = k/(k-1)$. Specifically,
\[ T_{\infty} \defeq \lim_{k \to \infty} \opt\left(\IP\left(k, \frac{k}{k-1}\right)\right)
\approx 1.691030206757254 \]
We get the same constant $T_{\infty}$ for $\mu = k/(k-2)$ and $\mu = k(k-2)/(k^2-3k+1)$.

\section{Simple Bounds}

When $k = 1$, $f_k(x) = \mu x$ for all $x \in [0, 1]$.
Therefore, $\opt(\LP(\mu, k)) = \mu$.

Let $I_j$ be the interval $(\frac{1}{j+1}, \frac{1}{j}]$ when $j \in [k-1]$
and $[0, 1/k]$ when $j = k$.

\begin{lemma}
\label{thm:fk-ratio}
$f_k(x) \le \max(\mu, 2)x$.
\end{lemma}
\begin{proof}
If $x \in I_k$, then $f_k(x)/x = \mu$.
If $x \in I_j$, where $j \in [k-1]$, then
\[ x > \frac{1}{j+1}
\implies \frac{f_k(x)}{x} = \frac{1/j}{x} < 1 + \frac{1}{j} \le 2  \qedhere \]
\end{proof}

\begin{lemma}
\label{thm:simple-bound}
$\mu \le \opt(\LP(k, \mu)) \le \max(\mu, 2)$.
\end{lemma}
\begin{proof}
For $X = [1/k: i \in [k]]$. Then $\Sum(X) = 1$ and $\Sum(f_k(X)) = \mu$.
Therefore, $\mu \le \opt(\LP(k, \mu))$.

For any vector $X$, by \cref{thm:fk-ratio}, $\Sum(f_k(X)) \le \max(\mu, 2)\Sum(X)$.
Therefore, $\opt(\LP(k, \mu)) \le \max(\mu, 2)$.
\end{proof}

\Cref{thm:simple-bound} implies that when $\mu \ge 2$, then $\opt(\LP(k, \mu)) = \mu$.

\section{Reduction to an Integer Program}

To find $\opt(\LP(k, \mu))$, we need to explore properties of
(near-)optimal solutions to $\LP(k, \mu)$.

\begin{observation}
\label{obs:dec-in-group}
Let $X$ be a set of numbers such that $\Sum(X) \le 1$ and $x \in X \cap I_j$, where $j \in [k-1]$.
If we replace $x$ by $x' = (1+\eps)/(j+1)$, then for sufficiently small $\eps$,
$x' \in I_j$ and $x' < x$. This decreases $\Sum(X)$, but $\Sum(f_k(X))$ remains the same.
\end{observation}

\begin{observation}
\label{obs:inc-sum}
For any set $X$ of numbers, if $\Sum(X) < 1$, then we can add numbers from $(0, 1/k]$
such that $\Sum(X)$ becomes 1. Doing this will increase $\Sum(f_k(X))$.
\end{observation}

Let $X$ be a near-optimal solution to $\LP(k, \mu)$.
Let $z_j \defeq |X \cap I_j|$.
Since the sum of numbers in $X - I_k$ should be at most 1, and each number in $I_j$ is
strictly larger than $1/(j+1)$, we get
\[ \sum_{j=1}^{k-1} \frac{z_j}{j+1} < 1 \]
By \cref{obs:dec-in-group}, we can assume that numbers in $X \cap I_j$ are slightly larger
than $1/(j+1)$, so by \cref{obs:inc-sum}, we get that $\Sum(X \cap I_k)$ is roughly equal to
\[ 1 - \sum_{j=1}^{k-1} \frac{z_j}{j+1} \]
Therefore, $\Sum(f_k(X))$ is roughly equal to
\[ \sum_{j=1}^{k-1} \frac{z_j}{j} + \mu\left(1 - \sum_{j=1}^{k-1} \frac{z_j}{j+1}\right)
= \mu + \sum_{j=1}^{k-1} z_j\left(\frac{1}{j} - \frac{\mu}{j+1}\right) \]
This suggests that we only need to determine $z_j$ for each $j \in [k-1]$
by solving the following integer LP (which we denote by $\IP(k, \mu)$)
and then we can get $X$ using \cref{obs:dec-in-group,obs:inc-sum}.
\[ \min_{z \in \mathbb{Z}_{\ge 0}^{k-1}}
\mu + \sum_{j=1}^{k-1} z_j\left(\frac{1}{j} - \frac{\mu}{j+1}\right)
\textrm{ where } \sum_{j=1}^{k-1} \frac{z_j}{j+1} < 1 \]
This is a sketch of why $\opt(\LP(k, \mu))$ and $\opt(\IP(k, \mu))$ should be equal;
we will soon present a formal proof.
For notational convenience, define $\score(z, k, \mu)$ and $\cost(z, k)$ to be
the objective value of $z$ and the LHS of the constraint, respectively, i.e.,
\begin{align*}
\score(z, k, \mu) &\defeq
\mu + \sum_{j=1}^{k-1} z_j\left(\frac{1}{j} - \frac{\mu}{j+1}\right)
& \cost(z, k) &\defeq \sum_{j=1}^{k-1} \frac{z_j}{j+1}
\end{align*}

\begin{lemma}
\label{thm:ip-time}
$\IP(k, \mu)$ has at most $k!$ feasible solutions.
Moreover, in $O(k!)$ time, we can list all feasible solutions to $\IP(k, \mu)$.
\end{lemma}
\begin{proof}
Let $z$ be a feasible solution to $\IP(k, \mu)$.
Then $z_j \le j$, otherwise the constraint in $\IP(k, \mu)$ will not be satisfied.
$z_j$ can take $j+1$ possible values, from $0$ to $j$.
Therefore, the total number of feasible solutions is at most $k!$.
To list all feasible solutions to $\IP(k, \mu)$, find all combinations of values of $z_j$
such that $0 \le z_j \le j$ and then check if $\cost(z, k) < 1$.
\end{proof}

\begin{proof}[Proof of \cref{thm:ip-eq-lp}]
We will show that $\opt(\LP(k, \mu)) = \opt(\IP(k, \mu))$.

Let $X$ be a feasible solution to $\LP(k, \mu)$, i.e. $\Sum(X) \le 1$.
Let $X_j$ be the numbers in $X$ that lie in the interval $I_j$, i.e. $X_j \defeq X \cap I_j$.
Let $z_j \defeq |X_j|$. Then $z$ is a feasible solution to $\IP(k, \mu)$, because
\[ \cost(z, k) = \sum_{j=1}^{k-1} \frac{z_j}{j+1} < \sum_{j=1}^{k-1} \sum(X_j)
= \Sum(X) - \Sum(X_k) \le 1 \]

Also,
\begin{align*}
\Sum(f_k(X)) &= \sum_{j=1}^{k-1} \frac{|X_j|}{j} + \mu \Sum(X_k)
\\ &\le \sum_{j=1}^{k-1} \frac{z_j}{j} + \mu \left( 1 - \sum_{j=1}^{k-1} \Sum(X_j)\right)
\\ &< \sum_{j=1}^{k-1} \frac{z_j}{j} + \mu \left( 1 - \sum_{j=1}^{k-1} \frac{z_j}{j+1} \right)
\\ &= \mu + \sum_{j=1}^{k-1} z_j\left(\frac{1}{j} - \frac{\mu}{j+1}\right)
\\ &= \score(z, k, \mu) \le \opt(\IP(k, \mu))
\end{align*}
Since $X$ can be made to be arbitrarily close to $\opt(\LP(k, \mu))$,
$\opt(\LP(k, \mu)) \le \opt(\IP(k, \mu))$.

Let $z$ be the optimal solution to $\IP(k, \mu)$.
Let $s \defeq \cost(z, k) < 1$.
Let $\eps > 0$ be a constant.
Let $X_j$ consist of $z_j$ copies of $(1+\eps)/(j+1)$. Then
\[ \Sum(X - X_k) = (1+\eps)\sum_{j=1}^{k-1} \frac{z_j}{j+1} = (1+\eps)s \]
We want $(1+\eps)s \le 1$ for $X$ to be feasible for $\LP(k, \mu)$,
so enforce the condition $\eps \le 1/s - 1$.

Choose $X_k \subseteq I_k$ such that $\Sum(X_k) = 1 - (1+\eps)s$.
This can be done by adding $1/k$ to $X_k$ while $\Sum(X) \le 1/k$,
and then choosing $1 - \Sum(X)$ as the last number to be added to $X_k$.
Now $\Sum(X) = 1$, so $X$ is feasible for $\LP(k, \mu)$.
\begin{align*}
\opt(\LP(k, \mu)) &\ge \Sum(f_k(X))
\\ &= \sum_{j=1}^{k-1} \frac{z_j}{j} + \mu \Sum(X_k)
\\ &= \sum_{j=1}^{k-1} \frac{z_j}{j} + \mu(1 - (1+\eps)s)
\tag{by definition of $s$}
\\ &= \sum_{j=1}^{k-1} \frac{z_j}{j}
   + \mu\left(1 - \sum_{j=1}^{k-1} \frac{z_j}{j+1}\right) - \mu\eps s
\\ &= \score(z, k, \mu) - \mu\eps s
\\ &> \opt(\IP(k, \mu)) - \mu\eps
\tag{since $s < 1$}
\end{align*}
Since for all $0 < \eps \le 1/s-1$,
$\opt(\LP(k, \mu)) > \opt(\IP(k, \mu)) - \mu\eps$,
we get that $\opt(\LP(k, \mu)) \ge \opt(\IP(k, \mu))$.
\end{proof}

\begin{table}[!ht]
\centering
\begin{tabular}{|c|r|r|r|}
\hline \diagbox{$k$}{$\mu$}
    & ${\displaystyle \frac{k}{k-1}}$
    & ${\displaystyle \frac{k}{k-2}}$
    & ${\displaystyle \frac{k(k-2)}{k^2-3k+1}}$ \\
\hline $2$ & $2$ & -- & -- \\
\hline $3$ & $\sfrac{7}{4} = \texttt{1.75000000}$ & $3$ & $3$ \\
\hline $4$ & $\sfrac{31}{18} = \texttt{1.72222222}$ & $2$ & $\sfrac{9}{5} = \texttt{1.80000000}$ \\
\hline $5$ & $\sfrac{41}{24} = \texttt{1.70833333}$ & $\sfrac{11}{6} = \texttt{1.83333333}$ & $\sfrac{19}{11} = \texttt{1.72727273}$ \\
\hline $6$ & $\sfrac{17}{10} = \texttt{1.70000000}$ & $\sfrac{7}{4} = \texttt{1.75000000}$ & $\sfrac{65}{38} = \texttt{1.71052632}$ \\
\hline $7$ & $\sfrac{61}{36} = \texttt{1.69444444}$ & $\sfrac{26}{15} = \texttt{1.73333333}$ & $\sfrac{148}{87} = \texttt{1.70114943}$ \\
\hline $8$ & $\sfrac{83}{49} = \texttt{1.69387755}$ & $\sfrac{31}{18} = \texttt{1.72222222}$ & $\sfrac{139}{82} = \texttt{1.69512195}$ \\
\hline $9$ & $\sfrac{569}{336} = \texttt{1.69345238}$ & $\sfrac{12}{7} = \texttt{1.71428571}$ & $\sfrac{559}{330} = \texttt{1.69393939}$ \\
\hline $10$ & $\sfrac{320}{189} = \texttt{1.69312169}$ & $\sfrac{41}{24} = \texttt{1.70833333}$ & $\sfrac{2525}{1491} = \texttt{1.69349430}$ \\
\hline $11$ & $\sfrac{237}{140} = \texttt{1.69285714}$ & $\sfrac{46}{27} = \texttt{1.70370370}$ & $\sfrac{6329}{3738} = \texttt{1.69315142}$ \\
\hline $12$ & $\sfrac{391}{231} = \texttt{1.69264069}$ & $\sfrac{17}{10} = \texttt{1.70000000}$ & $\sfrac{3875}{2289} = \texttt{1.69287899}$ \\
\hline \end{tabular}

\caption{Value of $\IP(k, \mu)$ for small $k$ and some important $\mu$,
computed using the algorithm of \cref{thm:ip-time}.}
\label{table:ip-values}
\end{table}

Let $z$ be a feasible solution to $\IP(k, \mu)$.
In $\score(z, k, \mu)$, the coefficient of $z_j$ is $(1/j - \mu/(j+1))$.
For all $j$, decrease $z_j$ to 0 iff $(1/j - \mu/(j+1))$ is non-positive.
Let $z'$ be the modified solution.
Then $z'$ is feasible for $\IP(k, \mu)$ and $\score(z', k, \mu) \ge \score(z, k, \mu)$.
Therefore, we can assume \wLoG{} that in any optimal solution $z$ to $\IP(k, \mu)$,
$z_j$ is 0 if $(1/j - \mu/(j+1))$ is non-positive.
Let
\[ m \defeq \argmax_{1 \le j \le k-1} \left(\frac{1}{j} - \frac{\mu}{j+1} > 0\right) \]
Note that $m$ is well-defined iff $\mu < 2$ and $k \ge 2$,
because then for $j=1$, $1/j - \mu/(j+1)$ is positive and $1 \le j \le k-1$.
\textbf{Henceforth, we will assume that $k \ge 2$ and $\mu < 2$.}

\begin{lemma}
\label{thm:m}
$m = \ceil{1/\max(\mu - 1, 1/k)} - 1$.
\end{lemma}
\begin{proof}
\begin{align*}
& j \le m
\\ &\iff j \le k-1 \wedge (1/j - \mu/(j+1) > 0)
\\ &\iff j < k \wedge 1/j > \mu - 1
\\ &\iff 1/j > \max(\mu-1, 1/k)
\\ &\iff j < 1/\max(\mu-1, 1/k)
\\ &\iff j \le \ceil{1/\max(\mu-1, 1/k)}-1
\qedhere \end{align*}
\end{proof}

\section{Greedy Algorithm and Harmonic Numbers}

Increasing $z_i$ by 1 increases $\score(z, k, \mu)$ by $(1/i - \mu/(i+1))$
and increases $\cost(z, k)$ by $1/(i+1)$.
The ratio of these increases is $1/i - (\mu-1)$, which we call \emph{bang-per-buck}.
Intuitively, to maximize $\score$, we should try to have large values of $z_i$
for small indices $i$, because that gives us a larger bang-per-buck.
This strategy suggests a greedy algorithm, i.e. start with $z = \lavec{0}$,
and repeatedly increase $z_i$ by 1 for the smallest index $i \le m$ such that
the $\cost(z, k)$ continues to be less than 1.

For sufficiently large $k$ and small $\mu$, a simple calculation
shows that the indices picked by the greedy algorithm are
1, 2, 6, 42, etc. We'll now formalize this pattern of numbers.

\begin{definition}[Harmonic number]
Define the $j\Th$ harmonic number $r_j$ as
\[ r_j \defeq \begin{cases} 1 & j = 1
\\ r_{j-1}(r_{j-1}+1) & j > 1 \end{cases} \]
\end{definition}

\begin{table}[!ht]
\centering
\begin{tabular}{|c|c|c|c|c|c|c|c|}
\hline
$j$ & 1 & 2 & 3 & 4 & 5 & 6 & 7 \\
\hline
$r_j$ & 1 & 2 & 6 & 42 & 1806 & 3263442 & 10650056950806 \\
\hline
\end{tabular}
\caption{First few harmonic numbers.}
\end{table}

\begin{lemma}
\label{thm:hn-series}
\[ \sum_{j=1}^t \frac{1}{r_j + 1} = 1 - \frac{1}{r_{t+1}} \]
\end{lemma}
\begin{proof}
\[ \frac{1}{r_j + 1} + \frac{1}{r_{j+1}}
= \frac{1}{r_j + 1} + \frac{1}{r_j(r_j + 1)}
= \frac{1}{r_j + 1}\left(1 + \frac{1}{r_j}\right)
= \frac{1}{r_j} \]
\[ \sum_{j=1}^t \frac{1}{r_j + 1}
= \sum_{j=1}^t \left( \frac{1}{r_j} - \frac{1}{r_{j+1}} \right)
= 1 - \frac{1}{r_{t+1}}  \qedhere \]
\end{proof}

Let the first $t$ indices picked by the greedy algorithm be
$R = \{r_1, r_2, \ldots, r_t\}$.
Let $z_i$ be 1 if $i \in R$ else 0.
By \cref{thm:hn-series}, $\cost(z, k) = \sum_{j=1}^t 1/(r_t + 1) = 1 - 1/r_{(t+1)}$.
So the next index to increase must be at least $r_{(t+1)}$,
and the greedy algorithm will pick it if $r_{(t+1)} \le m$.
This explains the pattern of numbers chosen by the greedy algorithm.

\subsection{Lower-Bounding \texorpdfstring{$\opt(\IP(k, \mu))$}{opt(IP(k,mu))}
Using the Greedy Algorithm}

Let $Q \defeq \argmax_{j \ge 1} (r_j \le m)$.
$Q$ is well-defined when $\mu < 2$ and $k \ge 2$, since then $m \ge 1$.
Note that $r_Q \le m \le k-1$.

For $t \ge 0$, define $S_t \defeq \sum_{j=1}^t 1/r_j$
and $R_t \defeq \{r_1, r_2, \ldots, r_t\}$.

\begin{definition}
Let $z^{(q)} \in \mathbb{Z}^{k-1}_{\ge 0}$ be a vector for which $z_i = 1$
when $i \in R_q$ and $z_i = 0$ otherwise.
\end{definition}
Note that $z^{(Q)}$ is the output of the greedy algorithm.
$z^{(q)}$ is feasible because
\[ \cost(z^{(q)}, k) = \sum_{i=1}^{k-1} \frac{z^{(q)}_i}{i+1}
= \sum_{j=1}^q \frac{1}{r_j+1} = 1 - \frac{1}{r_{q+1}} < 1 \]

\begin{lemma}
\label{thm:greedy-perf}
$\score(z^{(q)}, k, \mu) = S_{q+1} - (\mu-1)/r_{q+1}$.
\end{lemma}
\begin{proof}
\begin{align*}
\score(z^{(q)}, k, \mu) &= \mu +
    \sum_{i=1}^{k-1} z^{(q)}_i\left(\frac{1}{i} - \frac{\mu}{i+1}\right)
= \mu + \sum_{j=1}^q \left(\frac{1}{r_j} - \frac{\mu}{r_j + 1}\right)
\\ &= \mu + S_q - \mu\left(1 - \frac{1}{r_{q+1}}\right)
\tag{by \cref{thm:hn-series}}
\\ &= S_q + \frac{\mu}{r_{q+1}}
= S_{q+1} + \frac{\mu-1}{r_{q+1}}
\qedhere \end{align*}
\end{proof}

\begin{lemma}
\label{thm:greedy-is-sort-of-best}
For all $q \le Q-1$, $\score(z^{(q)}, k, \mu) < \score(z^{(Q)}, k, \mu)$.
\end{lemma}
\begin{proof}
For all $q \le Q-1$, $z_i^{(q)} \le z_i^{(q+1)}$
and $z^{(q)}_{r_{(q+1)}} = 0 < 1 = z^{(q+1)}_{r_{(q+1)}}$.
Therefore, $\score(z^{(q)}, k, \mu) < \score(z^{(q+1)}, k, \mu)$.
\end{proof}

\Cref{thm:greedy-perf} gives us a lower bound of $S_{Q+1} - (\mu - 1)/r_{Q+1}$
on $\opt(\IP(k, \mu))$.

\section{Upper Bound}

Let $z^*$ be an optimal solution to $\IP(k, \mu)$.

For $0 \le t \le Q$, let $R_t \defeq \{r_1, r_2, \ldots, r_t\}$
and let $P(t)$ be the predicate that $z_i > 0$ if $i \in R_t$.
Let $q \defeq \argmax_{0 \le t \le Q} P(t)$.
$q$ is well-defined since $P(0)$ is always true.

\begin{lemma}
\label{thm:opt-zeros}
For $i \in R_q$, $z^*_i = 1$.
If $z_i^* > 0$, then $i \in R_q$ or $i \ge r_{q+1}+1$.
\end{lemma}
\begin{proof}
Let $m' = \min(m, r_{q+1}-1)$.
\begin{align*}
1 &> \cost(z^*, k) \ge \sum_{i=1}^{m'} \frac{z_i^*}{i+1}
\\ &= \sum_{j=1}^{q} \frac{1}{r_j + 1} + \left(\sum_{i \in R_q} \frac{z_i^*-1}{i+1}
    + \sum_{i \in [m'] - R_q} \frac{z_i^*}{i+1} \right)
\\ &= \left(1 - \frac{1}{r_{q+1}}\right) + \left(\sum_{i \in R_q} \frac{z_i^*-1}{i+1}
    + \sum_{i \in [m'] - R_q} \frac{z_i^*}{i+1} \right)
\\ \implies & \frac{1}{r_{q+1}} > \sum_{i \in R_q} \frac{z_i^*-1}{i+1}
    + \sum_{i \in [m'] - R_q} \frac{z_i^*}{i+1}
\\ \implies & 1 > \sum_{i \in R_q} (z_i^*-1) + \sum_{i \in [m'] - R_q} z_i^*
\tag{since $i+1 \le m'+1 \le r_{q+1}$}
\end{align*}
By $P(q)$, we get $z_i^* = 1$ for $i \in R_q$ and $z_i^* = 0$ for $i \in [m'] - R_q$.

If $q = Q$, then $r_q \le m < r_{q+1}$, so $z_i^* > 0$ implies $i \in R_q$.
For $q < Q$, $P(q+1)$ is false, so $z_{r_{(q+1)}} = 0$.
$r_{q+1} \le r_Q \le m$, so $m' = r_{q+1}-1$
and for $i \in [r_{q+1}-1]-R_q$, $z_i^* = 0$.
Therefore, $z_i^* > 0$ implies $i \in R_q$ or $i \ge r_{q+1}+1$.
\end{proof}

\begin{lemma}
\label{thm:ub}
$\score(z^*, k, \mu) \le S_{Q+1} + \max(\mu-1, 0)/r_{Q+1}$.
\end{lemma}
\begin{proof}
\textbf{Case 1}: $r_{q+1} \ge m$.\\
Then by \cref{thm:opt-zeros}, $z_i^* > 0$ iff $i \in R_q$.
Therefore, $z^* = z^{(q)}$, so by \cref{thm:greedy-perf,thm:greedy-is-sort-of-best},
$\score(z^*, k, \mu) = S_{q+1} + (\mu-1)/r_{q+1}
\le S_{Q+1} + (\mu-1)/r_{Q+1}$.

\textbf{Case 2}: $r_{q+1} \le m-1$.\\
Then $q \le Q-1$ and by \cref{thm:opt-zeros},
$z_i^* = 1$ for $i \in R_q$ and $z_i^* = 0$ for $i \in [r_{q+1}] - R_q$.
\begin{align}
& \score(z^*, k, \mu) = \mu + \sum_{i=1}^{k-1} z_i^*\left(\frac{1}{i} - \frac{\mu}{i+1}\right)
\nonumber
\\ &\quad= \mu + \sum_{j=1}^q \left(\frac{1}{r_j} - \frac{\mu}{r_j+1}\right)
    + \sum_{i=r_{(q+1)}+1}^m z_i^*\left(\frac{1}{i} - \frac{\mu}{i+1}\right)
\nonumber
\\ &\quad= \mu + S_q - \mu\left(1 - \frac{1}{r_{q+1}}\right)
    + \sum_{i=r_{(q+1)}+1}^m \frac{z_i^*}{i+1}\left(\frac{1}{i} - (\mu-1)\right)
\tag{by \cref{thm:hn-series}}
\\ \label{eqn:score-ub} &\quad\le \left(S_{q+1} + \frac{\mu-1}{r_{q+1}}\right)
    + \left(\frac{1}{r_{q+1}+1} - (\mu-1)\right)\sum_{i=r_{(q+1)}+1}^m \frac{z_i^*}{i+1}
\end{align}
\begin{align*}
1 &> \cost(z^*, k) = \sum_{j=1}^q \frac{1}{r_j + 1} + \sum_{i=r_{(q+1)}+1}^m \frac{z_i^*}{i+1}
\\ &= \left(1 - \frac{1}{r_{q+1}}\right) + \sum_{i=r_{(q+1)}+1}^m \frac{z_i^*}{i+1}
\tag{by \cref{thm:hn-series}}
\end{align*}
\begin{equation}
\label{eqn:resid-ub}
\implies \sum_{i=r_{(q+1)}+1}^m \frac{z_i^*}{i+1} < \frac{1}{r_{q+1}}
\end{equation}
\begin{align}
& r_{q+1} + 1 \le m = \ceil{\frac{1}{\max(\mu-1, 1/k)}} - 1 < \frac{1}{\max(\mu-1, 1/k)}
\tag{by \cref{thm:m}}
\\ \label{eqn:coeff-positive} &\implies \frac{1}{r_{(q+1)}+1} > \max(\mu-1, 1/k) \ge \mu-1
\end{align}
By \cref{eqn:score-ub,eqn:resid-ub,eqn:coeff-positive}, we get
\[ \score(z^*, k, \mu) < \left(S_{q+1} + \frac{\mu-1}{r_{q+1}}\right)
+ \left(\frac{1}{r_{q+1}+1} - (\mu-1)\right)\frac{1}{r_{q+1}}
= S_{q+2} \le S_{Q+1} \]

On combining case 1 and case 2, we get
$\score(z^*, k, \mu) \le S_{Q+1} + \max(0, \mu-1)/r_{Q+1}$.
\end{proof}

\Cref{thm:greedy-perf} with $q = Q$ and \cref{thm:ub} together give us \cref{thm:exact}.

\section{Limiting Behavior}

\begin{lemma}
Let $\mu_1 < \mu_2$. Then $\opt(\IP(k, \mu_1)) < \opt(\IP(k, \mu_2))$.
\end{lemma}
\begin{proof}
Let $z^*$ be an optimal solution to $\IP(k, \mu_1)$. Then
\begin{align*}
& \opt(\IP(k, \mu_1)) = \score(z^*, k, \mu_1)
= \mu_1 + \sum_{i=1}^{k-1} z^*_i\left(\frac{1}{i} - \frac{\mu_1}{i+1}\right)
\\ &= \sum_{i=1}^{k-1} \frac{z^*_i}{i} - \mu_1(1 - \cost(z^*, k))
< \sum_{i=1}^{k-1} \frac{z^*_i}{i} - \mu_2(1 - \cost(z^*, k))
\\ &= \score(z^*, k, \mu_2) \le \opt(\IP(k, \mu_2))
\qedhere \end{align*}
\end{proof}

\begin{lemma}
Let $k_1 < k_2$. Then $\opt(\IP(k_1, \mu)) \le \opt(\IP(k_2, \mu))$.
\end{lemma}
\begin{proof}
Let $z^*$ be an optimal solution to $\IP(k_1, \mu)$.
Let $\zhat$ be the vector obtained by adding $(k_2 - k_1)$ zeros to $z^*$, i.e.
\[ \zhat_i = \begin{cases} z^*_i & i < k_1 \\ 0 & k_1 \le i < k_2 \end{cases} \]
Then $\cost(\zhat, k_2) = \cost(z^*, k_1) < 1$, so $\zhat$ is feasible for $\IP(k_2, \mu)$ and
\[ \opt(\IP(k_2, \mu)) \ge \score(\zhat, k_2, \mu)
= \score(z^*, k_1, \mu) = \opt(\IP(k_1, \mu))  \qedhere \]
\end{proof}

\begin{theorem}
\label{thm:tk-dec}
Let $k \ge 2$ and $\mu_k$ and $\mu_{k+1}$ be constants such that
\[ \frac{k+1}{k} \le \mu_{k+1} \le \mu_k
\textrm{ and } \frac{k}{k-1} \le \mu_k < 2
\textrm{ and } \frac{1}{\mu_{k+1}-1} - \frac{1}{\mu_k-1} \le 1. \]
Then $\opt(\IP(k, \mu_k)) \ge \opt(\IP(k+1, \mu_{k+1}))$.
\end{theorem}
\begin{proof} Let
\begin{align*}
T_k &\defeq \opt(\IP(k, \mu_k)) & T_{k+1} &\defeq \opt(\IP(k+1, \mu_{k+1}))
\\ m_k &\defeq \ceil{\frac{1}{\max(\mu_k-1, 1/k)}}-1
& m_{k+1} &\defeq \ceil{\frac{1}{\max(\mu_{k+1}-1, 1/(k+1))}}-1
\\ Q_k &\defeq \argmax_{j \ge 1} (r_j \le m_k) & Q_{k+1} &\defeq \argmax_{j \ge 1} (r_j \le m_{k+1})
\end{align*}
Then
\begin{align*}
& \mu_{k+1} \le \mu_k \textrm{ and } 1/(k+1) \le 1/k
\\ &\implies \max(\mu_{k+1}-1, 1/(k+1)) \le \max(\mu_k-1, 1/k)
\\ &\implies m_{k+1} \ge m_k
\\ &\implies Q_{k+1} \ge Q_k
\end{align*}
\begin{align*}
& \frac{k}{k-1} \le \mu_k
\\ &\implies \frac{1}{k-1} \le \mu_k - 1 = \max(\mu_k - 1, 1/k)
\\ &\implies k-1 \ge \frac{1}{\mu_k-1} = \frac{1}{\max(\mu_k - 1, 1/k)}
\\ &\implies k-2 \ge \ceil{\frac{1}{\mu_k-1}}-1 = m_k
\end{align*}
Similarly, we get $k-1 \ge \ceil{1/(\mu_{(k+1)}-1)}-1 = m_{k+1}$.

We are given that $1/(\mu_{k+1}-1) \le 1 + 1/(\mu_k-1)$. Hence,
\[ m_{k+1} = \ceil{\frac{1}{\mu_{k+1}-1}}-1
\le \ceil{1 + \frac{1}{\mu_k-1}} - 1
= \ceil{\frac{1}{\mu_k-1}} = m_k+1 \]
Therefore, $m_{k+1} \in \{m_k, m_k+1\}$.

\textbf{Case 1}: $Q_{k+1} = Q_k$. Then by \cref{thm:exact},
\[ T_k = S_{Q_k} + \frac{\mu_k}{r_{(Q_k)+1}}
\ge S_{Q_{(k+1)}} + \frac{\mu_{k+1}}{r_{Q_{(k+1)}+1}}
= T_{k+1} \]

\textbf{Case 2}: $Q_{k+1} \ge Q_k+1$. Then $m_{k+1} = m_k+1$.
By the definition of $Q_k$ and $Q_{k+1}$, we get
\begin{align*}
& r_{Q_k} \le m_k < r_{Q_k+1} \textrm{ and } r_{Q_{(k+1)}} \le m_{k+1}
\\ &\implies m_{k+1} = m_k + 1 \le r_{Q_k+1} \le r_{Q_{(k+1)}} \le m_{k+1}
\\ &\implies Q_{k+1} = Q_k + 1 \textrm{ and } r_{Q_{(k+1)}} = m_{k+1}
\end{align*}
Let $Q \defeq Q_{k+1} = Q_k + 1$.
For any $t \in \mathbb{Z}_{\ge 0}$, let $S_t \defeq \sum_{j=1}^t 1/r_j$.
Then by \cref{thm:exact},
\begin{align*}
& T_k - T_{k+1} = \left(S_Q + \frac{\mu_k-1}{r_Q}\right)
    - \left(S_Q + \frac{\mu_{k+1}}{r_{Q+1}}\right)
\\ &\implies r_Q(T_k - T_{k+1}) = (\mu_k-1) - \frac{\mu_{k+1}}{r_Q + 1}
\\ &\quad = (\mu_k-1) - \frac{\mu_{k+1}}{m_{k+1}+1}
\\ &\quad = (\mu_k-1) - \frac{\mu_{k+1}}{\ceil{1/(\mu_{(k+1)}-1)}}
\\ &\quad > (\mu_k-1) - \frac{\mu_{k+1}}{1 + 1/(\mu_{(k+1)}-1)}
\\ &\quad = (\mu_k-1) - (\mu_{k+1}-1) \ge 0
\end{align*}
Therefore, for both case 1 and case 2, $T_k \ge T_{k+1}$.
\end{proof}

\begin{corollary}
\label{thm:tk-dec-lee}
Let $\mu_k \defeq k/(k-1)$.
Then $\opt(\IP(k, \mu_k))$ is a non-increasing function of $k$.
\end{corollary}
\begin{corollary}
Let $\mu_k \defeq k/(k-2)$.
Then $\opt(\IP(k, \mu_k))$ is a non-increasing function of $k$.
\end{corollary}
\begin{corollary}
Let $\mu_k \defeq k(k-2)/(k^2-3k+1)$.
Then $\opt(\IP(k, \mu_k))$ is a non-increasing function of $k$.
\end{corollary}

Let $S_t \defeq \sum_{j=1}^t 1/r_j$ and $S_{\infty} \defeq \lim_{t \to \infty} S_t$.
Let $\mu_k \defeq k/(k-1)$. Let $T_k \defeq \opt(\IP(k, \mu_k))$.
Let $T_{\infty} \defeq \lim_{k \to \infty} T_k$.
By \cref{thm:exact}, we get that $T_{\infty} = S_{\infty}$.

$S_t$ is an increasing function of $t$.
By \cref{thm:tk-dec-lee}, we get that $T_k$ is a non-increasing function of $k$.
Therefore, $S_t \le S_{\infty} = T_{\infty} \le T_k$.
This will allow us to compute lower and upper bounds on $S_{\infty}$.

$S_{10} \approx 1.691030206757254$.
Let $k = r_9+2$. Then $m = \ceil{1/\max(\mu_k-1, k)}-1 = k-2 = r_9$ and $Q = 9$.
By \cref{thm:exact},
\[ T_k = S_{10} + \frac{\mu_k-1}{r_{10}}
= S_{10} + \frac{1}{r_{10}(r_9+1)}
\approx 1.691030206757254 \]
Therefore, $S_{\infty} \approx 1.691030206757254$.

\end{document}